\documentclass[11pt]{article}

\usepackage[T1]{fontenc}
\usepackage[utf8]{inputenc}
\usepackage{lmodern}
\usepackage{amsmath,amssymb,amsthm,amsfonts,mathtools}
\usepackage{bm}
\usepackage{bbm}
\usepackage{mathrsfs}
\usepackage{enumitem}
\usepackage[margin=1in]{geometry}
\usepackage{xcolor}
\usepackage{hyperref}
\usepackage{url}

\hypersetup{
  colorlinks=true,
  linkcolor=blue!60!black,
  citecolor=blue!60!black,
  urlcolor=blue!60!black
}

\newtheorem{definition}{Definition}
\newtheorem{lemma}{Lemma}

\newtheorem{theorem}{Theorem}
\newtheorem{corollary}{Corollary}
\theoremstyle{remark}

\title{Diminishing Returns in Expanding Generative Models and G\"odel--Tarski--L\"ob Limits}
\author{Angshul Majumdar, Indraprastha Institute of Information Technology, Delhi}
\date{}

\begin{document}

\maketitle

\begin{abstract}
Modern generative modelling systems are increasingly improved by expanding model capacity, training data, and computational resources. While empirical studies have documented such scaling behaviour across architectures including generative adversarial networks, variational autoencoders, transformer-based models, and diffusion models, the theoretical limits of capability growth in expanding generative systems remain poorly understood.

In this paper we develop a general task-space framework for analysing expanding generative reasoning systems. Each system induces a subset of a global task space representing the tasks it can successfully solve, and system capability is measured by the probability mass of this solved-task set under a fixed task distribution. Within this framework we prove a structural result showing that, under mild assumptions, the marginal improvement in solved tasks must converge to zero as system capacity increases. Thus expanding generative systems may continue to gain capability, but the probability mass of newly solvable tasks necessarily diminishes asymptotically.

We further provide a prediction-theoretic refinement based on complexity-weighted hypothesis classes inspired by algorithmic probability, yielding quantitative bounds on marginal improvement in prediction settings. Finally, we examine logical reasoning tasks and show that classical results from mathematical logic---including Rosser incompleteness, Tarski's undefinability theorem, and L\"ob's theorem---imply the persistence of unresolved logical tasks within sufficiently expressive reasoning systems.

Together these results provide a mathematical perspective on the asymptotic behaviour of expanding generative systems, showing that long-run capability growth is constrained both by diminishing marginal improvements in task coverage and by fundamental logical limitations on internal reasoning.
\end{abstract}


\section{Introduction}
\label{sec:intro}

Generative modelling studies how to represent, predict, and sample from complex data-generating processes. Over the past decade, a broad spectrum of architectures has been developed for this purpose, including generative adversarial networks \cite{Goodfellow2014}, variational autoencoders \cite{Kingma2014}, transformer-based autoregressive models \cite{Vaswani2017}, and diffusion-based generative models \cite{Ho2020,Rombach2022}. Although these approaches differ substantially in architecture, training objective, and output modality, they share a common empirical pattern: increasing model capacity, data, and computational resources often improves performance on generation and prediction tasks.

This observation has led to intense interest in the scaling behaviour of modern learning systems. Empirical studies have reported regular relationships between performance, model size, data, and compute \cite{Kaplan2020,Hoffmann2022}. However, most existing analyses are architecture-specific, empirical, or tied to particular loss functions. Much less is understood about the following more general question: if one considers an abstract family of \emph{expanding generative reasoning systems}, what mathematical law governs the long-run growth of the tasks such systems can solve?

This paper addresses that question through a task-space formulation. Instead of analysing a particular neural architecture, we consider a sequence of generative reasoning systems indexed by capacity. Each system induces a subset of a global task space consisting of the tasks it can solve successfully. A probability measure on the task space quantifies the prevalence of tasks in the environment, and the capability of a system is measured by the probability mass of its solved-task set. This leads to a natural notion of utility growth under expansion.

Within this framework we study the asymptotic behaviour of the utility sequence associated with an expanding family of systems. The central structural question is whether continued expansion can sustain non-vanishing marginal gains in solved tasks. Our first main result shows that, under mild assumptions, the answer is negative: if solved-task sets expand monotonically within a fixed task distribution, then the marginal improvement in utility must converge to zero. Thus expanding generative systems may continue to gain capability, but the probability mass of newly solvable tasks necessarily diminishes asymptotically.

The paper also develops a prediction-theoretic refinement of this result. For prediction tasks, we introduce a complexity-weighted hypothesis class inspired by Solomonoff induction and algorithmic information theory \cite{Solomonoff1964a,Solomonoff1964b,LiVitanyi2008,Hutter2005}. This yields an explicit quantitative bound on successive utility increments in terms of the tail mass of a complexity-truncated prior. In this sense, the general structural theorem is supplemented by a rate-type estimate in prediction settings.

A second theme of the paper concerns logical reasoning tasks internal to the system. When the reasoning component of a generative system is represented by a sufficiently expressive formal theory, classical incompleteness phenomena become unavoidable. Using Rosser's strengthening of G\"odel's first incompleteness theorem \cite{Rosser1936}, together with Tarski's undefinability theorem \cite{Tarski1936} and L\"ob's theorem on provability and reflection \cite{Lob1955}, we show that certain logical tasks remain unresolved within every sufficiently expressive system. These results imply that capability expansion is constrained not only by diminishing marginal improvement in utility, but also by persistent logical frontiers that cannot be eliminated from within the system itself.

The contribution of the paper is therefore twofold. First, it provides a general task-space theorem showing that expanding generative reasoning systems exhibit diminishing marginal improvement under mild structural assumptions. Second, it shows that when such systems perform internal formal reasoning, classical logical obstructions imply the persistence of unsolved reasoning tasks. Together these results provide a mathematically explicit account of why capability growth in expanding generative systems can continue while nevertheless facing both asymptotic diminishing returns and permanent logical limitations.

The remainder of the paper is organized as follows. Section~\ref{sec:model} introduces the formal model of expanding generative reasoning systems. Section~\ref{sec:prediction} develops the prediction-theoretic refinement based on complexity-weighted hypothesis classes. Section~\ref{sec:utility-dynamics} studies the dynamics of utility growth, and Section~\ref{sec:structure} establishes the structural properties needed for the main asymptotic theorem. Section~\ref{sec:diminishing} proves that marginal utility gains vanish asymptotically. Section~\ref{sec:logical_limits} then analyses logical limits on internal reasoning tasks. Section~\ref{sec:conclusion} concludes with a discussion of the broader implications of these results for the theory of generative systems.

\section{Formal Model of Expanding Generative Reasoning Systems}
\label{sec:model}

Recent progress in generative modelling has been driven by increasingly expressive model classes and larger computational resources. Architectures such as generative adversarial networks, variational autoencoders, transformer-based autoregressive models, and diffusion models illustrate this trend \cite{Goodfellow2014,Kingma2014,Vaswani2017,Ho2020}. Despite architectural differences, these systems share a common structure: a parameterized generative model, an internal inference or reasoning mechanism, and a procedure for producing outputs. In this section we introduce a formal abstraction that captures these common elements.

\subsection{Task Space}

Let $\Omega$ denote a measurable task space equipped with a probability measure $\mu$.

Each element $\omega \in \Omega$ represents a computational or predictive task posed to the system. Examples include predicting missing values in observed data, generating samples consistent with a target distribution, answering structured reasoning queries, and producing outputs satisfying specified constraints.

\textbf{Assumption 1 (Fixed task distribution).}
The task space $\Omega$ and probability measure $\mu$ are fixed and independent of the system index $n$.

\subsection{Generative Model Classes}

Let $X$ denote a measurable data domain and let $P$ denote the true but unknown data-generating distribution over $X$.

\begin{definition}[Generative model class]
A generative model class is a family of probability distributions
\[
\mathcal{P}_\Theta = \{p_\theta(x): \theta \in \Theta\},
\]
parameterized by a parameter space $\Theta$.
\end{definition}

Representative examples include adversarial generators \cite{Goodfellow2014}, latent-variable models \cite{Kingma2014}, autoregressive transformer architectures \cite{Vaswani2017}, and diffusion-based generative models \cite{Ho2020}.

\subsection{Generative Reasoning Systems}

\begin{definition}[Generative reasoning system]
A generative reasoning system is a tuple
\[
S=(\mathcal{P},T,I),
\]
where $\mathcal{P}$ is a generative model class, $T$ is an internal reasoning theory or inference framework, and $I$ is an inference procedure that produces outputs for tasks in $\Omega$.
\end{definition}

The inference procedure may involve probabilistic sampling, sequential generation, or symbolic reasoning over intermediate representations.

\subsection{Expanding System Families}

\begin{definition}[Expanding generative system family]
A sequence of generative reasoning systems
\[
S_1,S_2,\ldots
\]
is called expanding if
\[
\mathcal{P}_1 \subseteq \mathcal{P}_2 \subseteq \cdots
\qquad\text{and}\qquad
T_1 \subseteq T_2 \subseteq \cdots .
\]
\end{definition}

Increasing the index $n$ corresponds abstractly to enlarging model capacity, training data, algorithmic sophistication, or computational resources.

\subsection{Solved-Task Sets}

\begin{definition}[Solved-task set]
The solved-task set of system $S_n$ is
\[
\mathcal{A}_n
=
\{\omega \in \Omega : S_n \text{ produces a valid solution for } \omega\}.
\]
\end{definition}

The precise meaning of ``valid solution'' depends on the task type but may correspond, for example, to prediction error below a specified threshold, generation of samples from the correct distribution, or production of a logically valid answer.

\subsection{System Utility and Marginal Improvement}

\begin{definition}[Utility]
The utility of system $S_n$ is defined as
\[
U(n)=\mu(\mathcal{A}_n).
\]
\end{definition}

Thus $U(n)$ represents the probability that a randomly drawn task from $\Omega$ can be solved by system $S_n$.

The marginal improvement obtained by increasing system capacity from level $n$ to level $n+1$ is
\[
\Delta(n)=U(n+1)-U(n).
\]

Equivalently, if
\[
\mathcal{N}_{n+1}
=
\mathcal{A}_{n+1}\setminus \mathcal{A}_n
\]
denotes the novelty set of newly solvable tasks, then
\[
\Delta(n)=\mu(\mathcal{N}_{n+1}).
\]

\section{A Rigorous Prediction-Theoretic Refinement}
\label{sec:prediction}

The measure-theoretic theorem proved later does not require any special structure on the task space beyond monotonicity of the solved-task sets. However, for prediction tasks one can obtain a quantitative refinement using a complexity-weighted hypothesis class. This section provides that refinement in a form fully consistent with the task-space framework of Section~\ref{sec:model}.

\subsection{Prediction Tasks}

We consider the subclass of tasks in which a system observes a context $c \in \mathcal{C}$ and must predict an outcome $y \in \mathcal{Y}$, where $\mathcal{Y}$ is finite. A prediction task is therefore a pair
\[
\omega=(c,y)\in\Omega_{\mathrm{pred}}:=\mathcal{C}\times\mathcal{Y}.
\]
We assume that $\Omega_{\mathrm{pred}}\subseteq\Omega$.

Let $\pi$ be a probability measure on $\mathcal{C}$ describing the distribution of contexts. For each context $c$, the outcome is generated according to a conditional distribution on $\mathcal{Y}$.

\subsection{Hypothesis Class and Prior}

Let $\mathcal{H}$ be a countable class of hypotheses. Each $h\in\mathcal{H}$ specifies a conditional distribution
\[
q_h(\,\cdot \mid c)
\]
on $\mathcal{Y}$ for every $c\in\mathcal{C}$.

Let $K(h)$ denote the Kolmogorov complexity of a description of $h$ with respect to a fixed universal description language \cite{LiVitanyi2008}. Define a prior
\[
w_h=\frac{2^{-K(h)}}{Z},
\qquad
Z:=\sum_{g\in\mathcal{H}}2^{-K(g)}.
\]
By Kraft's inequality, $Z\le 1$ for any prefix-free coding scheme, so the prior is well defined after normalization \cite{LiVitanyi2008}. Define the complexity-truncated class
\[
\mathcal{H}_n:=\{h\in\mathcal{H}:K(h)\le n\},
\]
and its prior mass
\[
Z_n:=\sum_{h\in\mathcal{H}_n} w_h,
\qquad
\tau_n:=1-Z_n=\sum_{K(h)>n} w_h.
\]
Since $\sum_{h\in\mathcal{H}} w_h=1$, we have $\tau_n\downarrow 0$ as $n\to\infty$.

\subsection{Full and Truncated Predictive Mixtures}

The full predictive mixture is
\[
q(\,\cdot \mid c):=\sum_{h\in\mathcal{H}} w_h\,q_h(\,\cdot \mid c).
\]

The normalized truncated predictive mixture is
\[
q_n(\,\cdot \mid c)
:=
\sum_{h\in\mathcal{H}_n} \frac{w_h}{Z_n}\,q_h(\,\cdot \mid c),
\qquad Z_n>0.
\]

Define also the normalized tail mixture
\[
r_n(\,\cdot \mid c)
:=
\sum_{K(h)>n}\frac{w_h}{\tau_n}\,q_h(\,\cdot \mid c),
\qquad \tau_n>0.
\]

Then, for every context $c$,
\[
q(\,\cdot \mid c)=Z_n q_n(\,\cdot \mid c)+\tau_n r_n(\,\cdot \mid c).
\]

\subsection{Decision-Theoretic Risk}

Let $\mathcal{U}$ denote an action space. Let
\[
\ell:\mathcal{U}\times\mathcal{Y}\to[0,1]
\]
be a bounded measurable loss function.

For a predictive distribution $\rho(\,\cdot \mid c)$, define the Bayes risk at context $c$ by
\[
V(\rho,c):=\inf_{u\in\mathcal{U}}\sum_{y\in\mathcal{Y}}\ell(u,y)\rho(y\mid c).
\]

The corresponding context-averaged risk is
\[
R(\rho):=\int_{\mathcal{C}} V(\rho,c)\,\pi(dc).
\]

We use the total variation distance in its dual form:
\[
\|\rho-\rho'\|_{\mathrm{TV}}
:=
\sup_{|f|\le 1}
\left|
\sum_{y\in\mathcal{Y}} f(y)\rho(y)
-
\sum_{y\in\mathcal{Y}} f(y)\rho'(y)
\right|.
\]
On finite spaces this is equivalent to the usual half-$\ell_1$ definition.

\subsection{Stability of Bayes Risk Under Prior Truncation}

\begin{lemma}[Uniform perturbation bound]
\label{lem:tv-bound}
For every context $c\in\mathcal{C}$,
\[
\|q(\,\cdot \mid c)-q_n(\,\cdot \mid c)\|_{\mathrm{TV}}\le \tau_n.
\]
\end{lemma}

\begin{proof}
Using
\[
q(\,\cdot \mid c)=Z_n q_n(\,\cdot \mid c)+\tau_n r_n(\,\cdot \mid c),
\]
we obtain
\[
q(\,\cdot \mid c)-q_n(\,\cdot \mid c)
=
\tau_n\bigl(r_n(\,\cdot \mid c)-q_n(\,\cdot \mid c)\bigr).
\]
Hence
\[
\|q(\,\cdot \mid c)-q_n(\,\cdot \mid c)\|_{\mathrm{TV}}
=
\tau_n\|r_n(\,\cdot \mid c)-q_n(\,\cdot \mid c)\|_{\mathrm{TV}}
\le \tau_n,
\]
because the total variation distance between two probability distributions is at most $1$.
\end{proof}

\begin{lemma}[Risk perturbation bound]
\label{lem:risk-bound}
For every predictive distributions $\rho$ and $\rho'$ on $\mathcal{Y}$ and every context $c$,
\[
|V(\rho,c)-V(\rho',c)|\le \|\rho-\rho'\|_{\mathrm{TV}}.
\]
Consequently,
\[
|R(q)-R(q_n)|\le \tau_n.
\]
\end{lemma}

\begin{proof}
Fix $c$ and $u\in\mathcal{U}$. Since the function $y\mapsto \ell(u,y)$ takes values in $[0,1]$, the definition of total variation gives
\[
\left|
\sum_{y\in\mathcal{Y}}\ell(u,y)\rho(y)
-
\sum_{y\in\mathcal{Y}}\ell(u,y)\rho'(y)
\right|
\le
\|\rho-\rho'\|_{\mathrm{TV}}.
\]
Taking the infimum over $u$ on both sides yields
\[
|V(\rho,c)-V(\rho',c)|\le \|\rho-\rho'\|_{\mathrm{TV}}.
\]
Applying Lemma~\ref{lem:tv-bound} pointwise in $c$ and integrating over $\pi$ proves the second claim.
\end{proof}

\subsection{Quantitative Diminishing Returns}

Define the predictive utility of the truncated class by
\[
U_{\mathrm{pred}}(n):=-R(q_n).
\]

\begin{theorem}[Quantitative marginal-gain bound]
\label{thm:quant-dim}
For the utilities $U_{\mathrm{pred}}(n)=-R(q_n)$ defined above,
\[
|U_{\mathrm{pred}}(n+1)-U_{\mathrm{pred}}(n)|
\le
\tau_n+\tau_{n+1}
\le 2\tau_n.
\]
In particular,
\[
\lim_{n\to\infty}|U_{\mathrm{pred}}(n+1)-U_{\mathrm{pred}}(n)|=0.
\]
\end{theorem}

\begin{proof}
By Lemma~\ref{lem:risk-bound},
\[
|R(q_n)-R(q)|\le \tau_n,
\qquad
|R(q_{n+1})-R(q)|\le \tau_{n+1}.
\]
Therefore
\[
|R(q_{n+1})-R(q_n)|
\le
|R(q_{n+1})-R(q)|+|R(q)-R(q_n)|
\le
\tau_{n+1}+\tau_n.
\]
Since $U_{\mathrm{pred}}(n)=-R(q_n)$, the same bound holds for the utilities. Because $\tau_{n+1}\le \tau_n$, the bound by $2\tau_n$ follows immediately.
\end{proof}

Theorem~\ref{thm:quant-dim} gives a prediction-theoretic refinement of the general structural theorem proved later. The core structural result only requires nested solved-task sets and bounded utility, whereas the present section yields an explicit asymptotic bound in terms of the tail mass of the complexity prior.

\section{Utility Dynamics of Expanding Generative Systems}
\label{sec:utility-dynamics}

Sections~\ref{sec:model} and~\ref{sec:prediction} introduced the formal model of generative reasoning systems and a prediction-theoretic refinement based on complexity-bounded hypothesis classes. In this section we analyse how system utility evolves as system capacity increases. The purpose of this section is to express the growth of solved-task sets in terms of marginal gains and to prepare the structural argument leading to Theorem~\ref{thm:diminishing}. Related notions of incremental improvement and novelty sets appear naturally in algorithmic prediction and sequential decision settings \cite{Solomonoff1964a,Solomonoff1964b,Hutter2005,LiVitanyi2008}.

\subsection{Solved-Task Growth}

Recall from Section~\ref{sec:model} that the solved-task set of system $S_n$ is
\[
\mathcal{A}_n=
\{\omega\in\Omega:S_n \text{ produces a valid solution for } \omega\}.
\]

When the system expands from $S_n$ to $S_{n+1}$, additional tasks may become solvable. We define the set of newly solvable tasks as
\[
\mathcal{N}_{n+1}=
\mathcal{A}_{n+1}\setminus\mathcal{A}_n.
\]

We refer to $\mathcal{N}_{n+1}$ as the \emph{novelty set} of step $n+1$.

\subsection{Utility Decomposition}

Recall that system utility is defined as
\[
U(n)=\mu(\mathcal{A}_n).
\]

Because
\[
\mathcal{A}_{n+1}=\mathcal{A}_n\cup \mathcal{N}_{n+1},
\qquad
\mathcal{A}_n\cap \mathcal{N}_{n+1}=\varnothing,
\]
we obtain
\[
U(n+1)=U(n)+\mu(\mathcal{N}_{n+1}).
\]

Hence the marginal improvement defined in Section~\ref{sec:model} can be written as
\[
\Delta(n)=U(n+1)-U(n)=\mu(\mathcal{N}_{n+1}).
\]

This identity expresses marginal improvement entirely in terms of the probability mass of newly solvable tasks.

\subsection{Cumulative Improvement}

Iterating the above relation yields
\[
U(n)=U(1)+\sum_{k=1}^{n-1}\Delta(k).
\]

Thus the evolution of system utility can be interpreted as the cumulative contribution of successive novelty sets.

\subsection{Prediction-Theoretic Interpretation}

For prediction tasks considered in Section~\ref{sec:prediction}, system utility can be related to prediction risk. Let
\[
U_{\mathrm{pred}}(n)=-R(q_n)
\]
denote the predictive utility associated with the truncated mixture predictor $q_n$ defined in Section~\ref{sec:prediction}. The analysis there established that
\[
|U_{\mathrm{pred}}(n+1)-U_{\mathrm{pred}}(n)|\le 2\tau_n,
\]
where
\[
\tau_n=\sum_{K(h)>n}w_h
\]
is the complexity-prior tail mass. Such tail bounds are a standard consequence of universal prediction theory and algorithmic probability \cite{Solomonoff1964a,Solomonoff1964b,Hutter2005,LiVitanyi2008}. They provide a quantitative estimate of marginal improvement in the predictive setting.

\subsection{Structural Implications}

The decomposition
\[
U(n)=U(1)+\sum_{k=1}^{n-1}\Delta(k)
\]
shows that system utility grows through successive novelty sets. Accordingly, the asymptotic behaviour of the marginal gains $\Delta(n)$ determines how quickly system utility approaches its limiting value. In the next section we establish the structural conditions under which the solved-task sets form a nested sequence. Those conditions imply that the utility sequence is monotone and bounded, which in turn yields the main diminishing-returns theorem of Section~\ref{sec:diminishing}.

\section{Structural Properties of Expanding Generative Systems}
\label{sec:structure}

This section establishes the structural link between the formal model of Section~\ref{sec:model} and the asymptotic theorem proved in Section~\ref{sec:diminishing}. The key point is that, under a natural preservation assumption, expansion of model and reasoning capacity induces nested solved-task sets. This is the only structural property required for the general diminishing-marginal-improvement theorem. The prediction-theoretic analysis of Section~\ref{sec:prediction} provides an additional quantitative refinement on the prediction subspace $\Omega_{\mathrm{pred}}\subseteq\Omega$.

\subsection{Capability Preservation}

The expanding-family definition in Section~\ref{sec:model} ensures that the representational class and the internal reasoning framework both grow with the system index. To connect this formal expansion to the solved-task sets $\mathcal{A}_n$, we impose the following preservation assumption.

\textbf{Assumption 2 (Capability preservation).}
If a task $\omega\in\Omega$ can be validly solved by system $S_n$, then the same generative model and the same valid reasoning chain remain available in system $S_{n+1}$.

Assumption~2 formalizes a monotonic notion of system expansion: increasing capacity does not invalidate previously correct solutions. This assumption is structural rather than statistical; it does not require any specific architecture and is compatible with the abstract framework of Section~\ref{sec:model}.

\subsection{Nested Capability}

\begin{lemma}[Nested capability]
\label{lemma:nested}
Under Assumptions~1 and~2,
\[
\mathcal{A}_1\subseteq \mathcal{A}_2\subseteq \cdots .
\]
\end{lemma}

\begin{proof}
Fix $n\ge 1$ and let $\omega\in\mathcal{A}_n$. By definition of the solved-task set, system $S_n$ produces a valid solution for $\omega$.

Since $S_n=(\mathcal{P}_n,T_n,I_n)$ belongs to an expanding family, we have
\[
\mathcal{P}_n\subseteq \mathcal{P}_{n+1}
\qquad\text{and}\qquad
T_n\subseteq T_{n+1}.
\]
Therefore every generative model available at level $n$ remains representable at level $n+1$, and every reasoning step valid in $T_n$ remains valid in $T_{n+1}$.

By Assumption~2, the particular valid solution of $\omega$ available in $S_n$ remains available in $S_{n+1}$. Hence $\omega\in\mathcal{A}_{n+1}$.

Since $\omega\in\mathcal{A}_n$ was arbitrary, it follows that
\[
\mathcal{A}_n\subseteq \mathcal{A}_{n+1}.
\]
As $n$ was arbitrary, the full nested chain follows.
\end{proof}

\subsection{Consequences for Utility}

Recall from Section~\ref{sec:model} that
\[
U(n)=\mu(\mathcal{A}_n)
\qquad\text{and}\qquad
\Delta(n)=U(n+1)-U(n).
\]

Lemma~\ref{lemma:nested} implies immediately that the utility sequence is monotone non-decreasing:
\[
U(1)\le U(2)\le \cdots .
\]
Because $\mu$ is a probability measure, one also has
\[
0\le U(n)\le 1
\qquad\text{for all } n\ge 1.
\]
Thus the utility sequence is bounded and monotone. Section~\ref{sec:diminishing} uses exactly these two properties to prove that the marginal improvements $\Delta(n)$ must vanish asymptotically.

\subsection{Relation to the Prediction-Theoretic Refinement}

The argument above is purely structural and applies to the full task space $\Omega$. In contrast, Section~\ref{sec:prediction} restricts attention to the prediction subspace $\Omega_{\mathrm{pred}}$ and introduces a complexity-weighted hypothesis class together with truncated predictive mixtures. In that setting the prediction utility
\[
U_{\mathrm{pred}}(n)=-R(q_n)
\]
obeys the explicit quantitative bound of Theorem~\ref{thm:quant-dim}.

Accordingly, the present section supplies the monotonic-set property needed for the general theorem, whereas Section~\ref{sec:prediction} supplies a sharper, rate-type statement for prediction tasks.

\section{Diminishing Marginal Improvement}
\label{sec:diminishing}

We now derive the central structural theorem of the paper. The result is measure-theoretic in nature and depends only on the formal framework introduced in Section~\ref{sec:model} and the nested capability property established in Section~\ref{sec:structure}. The prediction-theoretic analysis of Section~\ref{sec:prediction} then provides a quantitative refinement for the special case of prediction tasks \cite{Solomonoff1964a,Solomonoff1964b,LiVitanyi2008,Hutter2005}.

\subsection{Utility and Marginal Gain}

Recall from Section~\ref{sec:model} that the solved-task set of system $S_n$ is
\[
\mathcal{A}_n
=
\{\omega\in\Omega:S_n \text{ produces a valid solution for } \omega\},
\]
and its utility is
\[
U(n)=\mu(\mathcal{A}_n).
\]

The marginal improvement obtained by expanding system capacity from level $n$ to level $n+1$ is
\[
\Delta(n)=U(n+1)-U(n).
\]

Equivalently, if
\[
\mathcal{N}_{n+1}
=
\mathcal{A}_{n+1}\setminus \mathcal{A}_n
\]
denotes the novelty set of newly solvable tasks, then
\[
\Delta(n)=\mu(\mathcal{N}_{n+1}).
\]

\subsection{Monotonicity and Boundedness}

By Lemma~\ref{lemma:nested}, under Assumptions~1 and~2 the solved-task sets are nested:
\[
\mathcal{A}_1\subseteq \mathcal{A}_2\subseteq \cdots .
\]

Therefore the utility sequence is monotone non-decreasing:
\[
U(1)\le U(2)\le \cdots .
\]

Moreover, since $\mu$ is a probability measure on $\Omega$, we have
\[
0\le U(n)\le 1
\qquad\text{for all } n\ge 1.
\]

Hence $\{U(n)\}_{n\ge 1}$ is a bounded monotone sequence. By the monotone convergence property for real sequences, the limit exists \cite{Royden2010,Folland1999}.

\begin{lemma}
\label{lem:utility-converges}
The limit
\[
U_\infty:=\lim_{n\to\infty}U(n)
\]
exists.
\end{lemma}

\begin{proof}
Since $\{U(n)\}_{n\ge 1}$ is monotone non-decreasing and bounded above by $1$, it converges \cite{Royden2010,Folland1999}.
\end{proof}

\subsection{Telescoping Representation}

The cumulative gain up to level $n$ admits the telescoping representation
\[
U(n)-U(1)=\sum_{k=1}^{n-1}\Delta(k).
\]

This identity follows immediately from the definition $\Delta(k)=U(k+1)-U(k)$ and will be used in the proof below.

\subsection{Main Structural Theorem}

\begin{theorem}[Diminishing marginal improvement]
\label{thm:diminishing}
Let $\{S_n\}_{n\ge 1}$ be an expanding family of generative reasoning systems satisfying Assumptions~1 and~2. Then
\[
\lim_{n\to\infty}\Delta(n)=0.
\]
\end{theorem}

\begin{proof}
By Lemma~\ref{lem:utility-converges}, the limit
\[
U_\infty:=\lim_{n\to\infty}U(n)
\]
exists and is finite.

Using the telescoping representation,
\[
\sum_{k=1}^{n-1}\Delta(k)=U(n)-U(1).
\]
Taking $n\to\infty$ yields
\[
\sum_{k=1}^{\infty}\Delta(k)=U_\infty-U(1)<\infty.
\]

Since each $\Delta(k)\ge 0$, convergence of the series implies
\[
\lim_{k\to\infty}\Delta(k)=0.
\]
This proves the claim.
\end{proof}

\subsection{Equivalent Set-Theoretic Formulation}

Theorem~\ref{thm:diminishing} may also be written directly in terms of novelty sets. Since
\[
\Delta(n)=\mu(\mathcal{N}_{n+1}),
\]
Theorem~\ref{thm:diminishing} is equivalent to
\[
\lim_{n\to\infty}\mu(\mathcal{N}_{n+1})=0.
\]

Thus, although expanding system capacity may continue to enlarge the solved-task sets, the probability mass of newly solvable tasks must vanish asymptotically.

\subsection{Relation to the Prediction-Theoretic Refinement}

Theorem~\ref{thm:diminishing} is purely structural: it requires only a fixed task distribution, capability preservation, and nested solved-task sets. It does not depend on any particular learning architecture or on any special representation of the task space.

In contrast, Section~\ref{sec:prediction} considered prediction tasks equipped with a complexity-weighted hypothesis class and established an explicit quantitative bound in terms of the tail mass
\[
\tau_n=\sum_{K(h)>n} w_h.
\]

That analysis shows that, in the predictive setting,
\[
|U_{\mathrm{pred}}(n+1)-U_{\mathrm{pred}}(n)|\le 2\tau_n,
\]
which is consistent with the qualitative conclusion of Theorem~\ref{thm:diminishing} because $\tau_n\to 0$ \cite{Solomonoff1964a,Solomonoff1964b,LiVitanyi2008,Hutter2005}.

Accordingly, Theorem~\ref{thm:diminishing} provides the general asymptotic statement, while Section~\ref{sec:prediction} supplies a sharper rate-type refinement for prediction tasks.

\section{Logical Limits of Expanding Reasoning Systems}
\label{sec:logical_limits}

Sections~\ref{sec:model}--\ref{sec:diminishing} established that expanding generative reasoning systems exhibit diminishing marginal improvement in the probability mass of solvable tasks. We now analyse the logical limitations that arise when such systems perform internal reasoning about statements expressed in a formal language. Unlike the structural theorem of Section~\ref{sec:diminishing}, the present section is not used to prove diminishing marginal improvement. Rather, it shows that even if system capacity continues to expand, there remains a persistent class of logical tasks that cannot be completely resolved from within the system itself \cite{Godel1931,Rosser1936,Tarski1936,Lob1955,Boolos2002,Smorynski1977}.

\subsection{Logical Task Space}

Let $\mathcal{L}_{\mathrm{arith}}$ be a formal language extending first-order arithmetic. Define the logical task space
\[
\Omega_{\mathrm{logic}}
=
\{\varphi : \varphi \text{ is a sentence of } \mathcal{L}_{\mathrm{arith}}\}.
\]
We assume that
\[
\Omega_{\mathrm{logic}}\subseteq \Omega,
\]
so that logical reasoning tasks form a distinguished subclass of the global task space introduced in Section~\ref{sec:model}.

\subsection{Logical Solvability}

Let $T_n$ denote the reasoning component of system $S_n$, viewed as a formal theory in the language $\mathcal{L}_{\mathrm{arith}}$.

\begin{definition}[Logical task solvability]
A sentence $\varphi\in\Omega_{\mathrm{logic}}$ is said to be solved by system $S_n$ if
\[
T_n\vdash \varphi
\qquad\text{or}\qquad
T_n\vdash \neg\varphi .
\]
\end{definition}

Accordingly, the solved logical task set is
\[
\mathcal{A}_n^{\mathrm{logic}}
=
\{\varphi\in\Omega_{\mathrm{logic}}: T_n\vdash \varphi \text{ or } T_n\vdash \neg\varphi\}.
\]
By construction,
\[
\mathcal{A}_n^{\mathrm{logic}}\subseteq \mathcal{A}_n.
\]

\subsection{Assumptions on the Reasoning Theory}

We assume that each theory $T_n$ satisfies the following standard conditions:
\begin{itemize}
\item $T_n$ is recursively axiomatizable,
\item $T_n$ is consistent,
\item $T_n$ extends a sufficient fragment of arithmetic.
\end{itemize}
These assumptions are standard in the study of formal reasoning systems and incompleteness phenomena \cite{Boolos2002,Smorynski1977}.

\subsection{Rosser Incompleteness}

To avoid stronger soundness assumptions, we use Rosser's strengthening of Gödel's first incompleteness theorem.

\begin{theorem}[Rosser incompleteness theorem]
\label{thm:rosser}
Let $T$ be a consistent recursively axiomatizable theory extending a sufficient fragment of arithmetic. Then there exists a sentence $\rho_T$ such that
\[
T\nvdash \rho_T
\qquad\text{and}\qquad
T\nvdash \neg\rho_T .
\]
\end{theorem}

This theorem implies that each theory $T_n$ leaves some logical tasks undecided \cite{Rosser1936,Boolos2002}.

\begin{corollary}
\label{cor:rosser-unsolved}
For every system $S_n$, there exists a sentence
\[
\rho_n\in\Omega_{\mathrm{logic}}
\]
such that
\[
\rho_n\notin \mathcal{A}_n^{\mathrm{logic}}.
\]
Hence
\[
\rho_n\notin \mathcal{A}_n.
\]
\end{corollary}

\begin{proof}
Apply Theorem~\ref{thm:rosser} to $T_n$. The resulting sentence $\rho_n$ is neither provable nor refutable in $T_n$, so by definition it does not belong to $\mathcal{A}_n^{\mathrm{logic}}$. Since $\mathcal{A}_n^{\mathrm{logic}}\subseteq \mathcal{A}_n$, the final claim follows.
\end{proof}

\subsection{Tarski Undefinability}

A second obstruction concerns truth evaluation within the language itself.

\begin{theorem}[Tarski undefinability theorem]
\label{thm:tarski}
Let $\mathcal{L}_{\mathrm{arith}}$ be sufficiently expressive to formalize arithmetic. Then arithmetic truth for sentences of $\mathcal{L}_{\mathrm{arith}}$ is not definable in $\mathcal{L}_{\mathrm{arith}}$ itself.
\end{theorem}

Equivalently, there is no formula $\mathrm{True}(x)$ in $\mathcal{L}_{\mathrm{arith}}$ such that for every sentence $\varphi$ of $\mathcal{L}_{\mathrm{arith}}$,
\[
\mathbb{N}\models \mathrm{True}(\ulcorner \varphi \urcorner)
\qquad\Longleftrightarrow\qquad
\mathbb{N}\models \varphi .
\]

Thus no system whose internal reasoning is expressed in $\mathcal{L}_{\mathrm{arith}}$ can define a complete internal truth predicate for all of its own arithmetic sentences \cite{Tarski1936,Smorynski1977}.

\subsection{Löb and Uniform Reflection}

We next consider self-verification tasks. Let $\mathrm{Prov}_{T_n}(x)$ denote a standard arithmetized provability predicate for $T_n$.

\begin{theorem}[Löb's theorem]
\label{thm:lob}
For any sentence $\varphi$ of the language of $T_n$,
\[
T_n\vdash \bigl(\mathrm{Prov}_{T_n}(\ulcorner \varphi \urcorner)\rightarrow \varphi\bigr)
\qquad\Longrightarrow\qquad
T_n\vdash \varphi .
\]
\end{theorem}

Löb's theorem implies that unrestricted internal reflection principles cannot be established wholesale within $T_n$ unless the corresponding reflected statements are already provable \cite{Lob1955,Smorynski1977,Boolos2002}.

In particular, tasks asking the system to certify a uniform schema of the form
\[
\forall \varphi\,\bigl(\mathrm{Prov}_{T_n}(\ulcorner \varphi \urcorner)\rightarrow \varphi\bigr)
\]
cannot in general be solved internally without passing to a substantially stronger theory. Hence unrestricted self-verification tasks remain outside the scope of $\mathcal{A}_n^{\mathrm{logic}}$.

\subsection{Persistence of Unsolved Logical Tasks}

\begin{theorem}[Persistence of unsolved logical tasks]
\label{thm:persistence}
Let $\{S_n\}_{n\ge 1}$ be an expanding family of generative reasoning systems whose reasoning components $\{T_n\}_{n\ge 1}$ satisfy the assumptions above. Then for every $n$ there exists a task
\[
\omega_n\in\Omega_{\mathrm{logic}}
\]
such that
\[
\omega_n\notin \mathcal{A}_n.
\]
Moreover, the source of unsolvability may arise from incompleteness, undefinability of arithmetic truth, or limitations on unrestricted internal reflection.
\end{theorem}

\begin{proof}
By Corollary~\ref{cor:rosser-unsolved}, for each $n$ there exists a sentence $\rho_n\in\Omega_{\mathrm{logic}}$ such that $\rho_n\notin\mathcal{A}_n$. This already proves the first claim.

The further statement follows from the nature of the three logical obstructions. Rosser incompleteness produces undecidable sentences. Tarski's theorem shows that complete internal truth-evaluation tasks for arithmetic sentences cannot be captured in the same language. Löb's theorem shows that unrestricted internal reflection cannot be uniformly established without already proving the reflected statements. Hence logical unsolvability persists in multiple formally distinct ways.
\end{proof}

\subsection{Interpretation}

Theorem~\ref{thm:persistence} complements the diminishing marginal improvement theorem of Section~\ref{sec:diminishing}. Section~\ref{sec:diminishing} showed that the probability mass of newly solvable tasks must vanish asymptotically under the structural assumptions of the paper. The present section shows that, for reasoning tasks internal to sufficiently expressive formal systems, expansion does not eliminate all unsolved tasks. Thus the limiting behaviour of expanding generative reasoning systems is constrained not only by diminishing returns in utility growth, but also by persistent logical frontiers.

\section{Discussion and Conclusion}
\label{sec:conclusion}

This paper developed a task-space framework for analysing the long-run behaviour of expanding generative reasoning systems. Instead of focusing on particular neural architectures or training objectives, the analysis treated generative systems abstractly as mechanisms that attempt to solve tasks drawn from an underlying probability distribution over a global task space. Within this formulation, system capability is measured by the probability mass of the solved-task set associated with a given system instance.

The central structural result of the paper shows that, under mild assumptions, the marginal utility gains of an expanding sequence of systems must converge to zero. Intuitively, if successive systems solve increasingly larger subsets of tasks while the underlying task distribution remains fixed, the probability mass of newly solvable tasks necessarily decreases. Consequently, although capability may continue to improve indefinitely, the incremental gains in solved tasks diminish asymptotically.

This result is intentionally architecture-agnostic. The analysis does not depend on specific model classes such as generative adversarial networks, variational autoencoders, transformer-based language models, or diffusion-based generative models \cite{Goodfellow2014,Kingma2014,Vaswani2017,Ho2020,Rombach2022}. Instead, the theorem applies to any expanding family of generative reasoning systems whose solved-task sets evolve within a fixed probabilistic task environment. In this sense, the result provides a general mathematical perspective on capability growth that complements empirical scaling analyses of modern machine learning systems \cite{Kaplan2020,Hoffmann2022}.

The prediction-theoretic refinement presented earlier further illustrates how diminishing improvements arise when model classes expand according to complexity constraints. When predictive hypotheses are weighted according to algorithmic complexity, successive truncations of the hypothesis class yield bounds on marginal utility increments governed by the tail mass of the complexity prior \cite{Solomonoff1964a,Solomonoff1964b,LiVitanyi2008,Hutter2005}. This provides a concrete quantitative mechanism through which diminishing improvements emerge in prediction settings.

Beyond structural limits on capability growth, the paper also examined logical constraints on reasoning systems that operate over formal languages. When the reasoning component of a generative system is represented by a sufficiently expressive formal theory, classical incompleteness phenomena arise. Rosser's strengthening of Gödel's incompleteness theorem guarantees the existence of undecidable sentences in any consistent recursively axiomatizable extension of arithmetic \cite{Rosser1936}. Tarski's undefinability theorem shows that arithmetic truth cannot be defined within the same language \cite{Tarski1936}. Löb's theorem further constrains the ability of systems to establish unrestricted internal reflection principles \cite{Lob1955}. Together, these results imply that certain reasoning tasks remain unresolved within every sufficiently expressive formal system \cite{Boolos2002,Smorynski1977}.

The interaction between these two perspectives---diminishing marginal improvement and persistent logical limitations---suggests a broader view of the asymptotic behaviour of expanding generative systems. On the one hand, capability growth is constrained probabilistically: the probability mass of newly solvable tasks decreases as systems cover larger portions of the task space. On the other hand, logical frontiers remain: even very powerful reasoning systems cannot internally resolve all formally expressible reasoning tasks. These phenomena operate at different conceptual levels but jointly limit the extent to which generative systems can achieve complete task solvability.

Several directions for future work remain open. One natural extension is to study more refined models of task distributions in which the environment itself evolves as systems become more capable. Another direction concerns explicit quantitative rates of diminishing improvement under different structural assumptions on hypothesis classes or task-generation processes. Finally, it would be interesting to explore connections between the present framework and recent theoretical work on scaling laws, sample complexity, and generalization in large learning systems.

In summary, the framework developed in this paper provides a mathematically explicit lens through which to view the long-term behaviour of expanding generative reasoning systems. The results indicate that capability growth may continue while nevertheless encountering two fundamental constraints: diminishing marginal improvements in solved tasks and persistent logical limitations on internal reasoning. These principles suggest that the asymptotic frontier of generative systems is shaped simultaneously by probabilistic structure in the task environment and by deep logical properties of formal reasoning.

\bibliographystyle{plain}
\bibliography{global}

\end{document}